%% file: Main.tex
\documentclass[a4paper,UKenglish,cleveref, autoref, thm-restate]{lipics-v2021}
\nolinenumbers

\title{Nested Sequents for First-Order Modal Logics via Reachability Rules} 

\author{Tim S. Lyon}{Computational Logic Group, Institute of Artificial Intelligence, Technische Universit\"at Dresden, Germany \and \url{https://sites.google.com/view/timlyon}}{timothy_stephen.lyon@tu-dresden.de}{https://orcid.org/0000-0003-3214-0828}{}

\authorrunning{T.\,S. Lyon} 

\Copyright{ } 



\category{} 

\relatedversion{} 

\funding{Work supported by the European Research Council (ERC) Consolidator Grant 771779 \textit{A Grand Unified Theory of Decidability in Logic-Based Knowledge Representation} (DeciGUT).}

\EventEditors{John Q. Open and Joan R. Access}
\EventNoEds{2}
\EventLongTitle{42nd Conference on Very Important Topics (CVIT 2016)}
\EventShortTitle{CVIT 2016}
\EventAcronym{CVIT}
\EventYear{2016}
\EventDate{December 24--27, 2016}
\EventLocation{Little Whinging, United Kingdom}
\EventLogo{}
\SeriesVolume{42}
\ArticleNo{23}

\input{notation.tex}

\begin{document}

\maketitle

\begin{abstract}
We introduce the first cut-free nested sequent systems for first-order modal logics that admit increasing, decreasing, constant, and empty domains along with so-called general path conditions and seriality. We obtain such systems by means of two devices: \emph{reachability rules} and \emph{structural refinement}. Regarding the former device, we introduce reachability rules as special logical rules parameterized with formal grammars (viz. types of semi-Thue systems) that operate by propagating formulae and/or checking if data exists along certain paths within a nested sequent, where paths are encoded as strings generated by a parameterizing grammar. Regarding the latter device, structural refinement is a relatively new methodology used to extract nested sequent systems from labeled systems (which are ultimately obtained from a semantics) by means of eliminating structural/relational rules, introducing reachability rules, and then carrying out a notational translation. We therefore demonstrate how this method can be extended to the setting of first-order modal logics, and expose how reachability rules naturally arise from applying this method.
\end{abstract}


\section{Introduction}

\input{introduction.tex}

\section{Preliminaries}\label{sec:log-prelims-I}

\subsection{First-order Modal Logics}\label{subsec:fo-modal-logics}

\input{prelims.tex}

\subsection{Grammar-theoretic Preliminaries}

\input{grammar-prelims}

\section{Labeled Sequent Systems}\label{sec:labeled-calc}

\input{labelled-calc.tex}

\section{Structural Refinement}\label{sec:struc-refinement}

\input{struc-refinement.tex}

\section{Nested Sequent Systems}\label{sec:nested-calculi}

\input{nested-systems.tex}

\section{Conclusion}\label{sec:conclusion}

\input{conclusion.tex}

\bibliography{bibliography}




\end{document}

%% file: notation.tex
\renewcommand{\phi}{\varphi}

\usepackage{amssymb, amsmath}
\usepackage{synttree, bussproofs,scrextend, stmaryrd, rotating, xcolor}
\usepackage{multicol}
\usepackage{comment}
\usepackage{hyperref}
\usepackage[all]{xy}
\usepackage{changepage,enumerate,proof,upgreek,relsize,trfsigns}

\usepackage{mathtools} 

\newcommand\new[0]{\reflectbox{\ensuremath{\mathsf{N}}}}

\newcommand{\ifandonlyif}{iff }
\newcommand{\iffi}{\textit{iff} }
\newcommand{\etc}{$\ldots$ }
\newcommand{\dfn}{Definition}
\newcommand{\thm}{Theorem}
\newcommand{\lem}{Lemma}

\newcommand{\sect}{Section}
\newcommand{\cptr}{Chapter}

\newcommand{\fig}{Figure}

\newcommand{\rmk}{Remark}

\newcommand{\albet}{\Upsigma}

\newcommand{\cate}{}
\newcommand{\sfour}{\mathrm{S4}}
\newcommand{\sfive}{\mathrm{S5}}
\newcommand{\langsfour}{L_{\sfour}}
\newcommand{\langsfive}{L_{\sfive}}

\newcommand{\concat}{\cdot}
\newcommand{\thuesys}{S}
\newcommand{\g}{\mathrm{S}}

\newcommand{\glang}{L_{\g}}
\newcommand{\pto}{\longrightarrow}

\newcommand{\dr}{\pto^{*}_{\g}}
\newcommand{\fd}{\Diamond}
\newcommand{\bd}{\blacklozenge}
\newcommand{\ques}{\langle ? \rangle}
\newcommand{\albetstr}{\albet^{*}}
\newcommand{\stra}{s}
\newcommand{\strb}{t}
\newcommand{\conv}[1]{{#1}^{-1}}
\newcommand{\strc}{r}
\newcommand{\empstr}{\varepsilon}

\newcommand{\rel}{\mathcal{R}}

\newcommand{\empseq}{\emptyset}
\newcommand{\lab}{\mathrm{Lab}}

\newcommand{\vx}{\mathrm{X}}
\newcommand{\vy}{\mathrm{Y}}

\newcommand{\pred}{\Upphi}
\newcommand{\var}{\mathrm{Var}}

\newcommand{\imp}{\rightarrow}

\newcommand{\lang}{\mathcal{L}}

\newcommand{\ns}{\Phi}
\newcommand{\nsii}{\Psi}
\newcommand{\nant}{\Gamma}
\newcommand{\nantii}{\Sigma}
\newcommand{\ncon}{\Delta}
\newcommand{\nconii}{\Pi}

\newcommand{\bl}{[}
\newcommand{\br}{]}

\newcommand{\hol}{\{}
\newcommand{\hor}{\}}
\newcommand{\sar}{\vdash}

\newcommand{\prgr}[1]{PG(#1)}
\newcommand{\prgrdom}{V}
\newcommand{\prgredges}{E}
\newcommand{\ppath}{\pi}
\newcommand{\emppath}{\varepsilon}

\newcommand{\id}{(ax)}
\newcommand{\ax}{(ax)}

\newcommand{\disr}{(\lor_{r})}

\newcommand{\negr}{(\neg_{r})}
\newcommand{\negl}{(\neg_{l})}

\newcommand{\existsr}{(\exists_{r})}

\newcommand{\disl}{(\lor_{l})}

\newcommand{\existsl}{(\exists_{l})}

\newcommand{\domr}{(ds)}

\newcommand{\ru}{(r)}

\newcommand{\ntr}{\mathfrak{N}}
\newcommand{\ltr}{\mathfrak{L}}
\newcommand{\seqcomp}{\otimes}

\newcommand{\D}{(d)}

\newcommand{\dia}{\Diamond}

\newcommand{\diar}{(\dia_{r})}
\newcommand{\dial}{(\dia_{l})}

\newcommand{\sa}{(g_{n,k})}

\newcommand{\prdia}{(p_{\dia})}

\newcommand{\asmt}{\mu}
\newcommand{\qk}{\mathsf{QK}}
\newcommand{\lqk}{\mathsf{G3}\qk}
\newcommand{\lql}{\mathsf{G3}\ql}
\newcommand{\qll}{\ql\mathsf{L}}
\newcommand{\idom}{(id)}
\newcommand{\ddom}{(dd)}
\newcommand{\cdom}{(cd)}
\newcommand{\ned}{(nd)}

\newcommand{\prexi}{(s_{\exists}^{1})}
\newcommand{\prexii}{(s_{\exists}^{2})}

\newcommand{\fcset}{\mathcal{C}}
\newcommand{\gpset}{\mathcal{G}}
\newcommand{\ql}{\mathsf{QK}(\fcset)}
\newcommand{\gpc}{G}
\newcommand{\serc}{S}
\newcommand{\idc}{I_{d}}
\newcommand{\ddc}{D_{d}}
\newcommand{\cdc}{C_{d}}
\newcommand{\nedc}{N_{d}}
\newcommand{\nql}{\mathsf{N}\ql}
\newcommand{\nqk}{\mathsf{N}\qk}

\newcommand{\I}[1]{I(#1)}

%% file: introduction.tex
 Analytic proof systems consist of sets of inference rules that stepwise (de)compose logical theorems. Such proof systems not only facilitate reasoning within a given logic, but have proven useful for establishing (meta-)logical properties of logics as well as for developing automated reasoning tools. One of the most prominent formalisms for constructing analytic proof systems is Gentzen's sequent calculus~\cite{Gen35a,Gen35b}, which utilizes (pairs of) multisets of formulae, called \emph{sequents}, to carry out proofs of logical theorems. Since its introduction, Gentzen's formalism has been generalized to supply ever more expressive logics with analytic proof systems. Such extensions include \emph{display calculi}~\cite{Bel82}, \emph{hypersequent calculi}~\cite{Avr96}, \emph{labeled calculi}~\cite{Sim94,Vig00}, and \emph{nested calculi}~\cite{Bul92,Kas94}, all of which have been exploited to establish meaningful results such as decidability~\cite{Bru09,TiuIanGor12}, interpolation~\cite{FitKuz15,LyoTiuGorClo20}, complexity hardness~\cite{LyoAlv22}, and automated counter-model extraction~\cite{LyoBer19,TiuIanGor12}. 
 In this paper, we will be concerned with the labeled and nested formalisms.
 
 Labeled sequent calculi generalize Gentzen's formalism by integrating semantic elements of a logic into the the syntax of sequents (cf.~\cite{Kan57}). As such, labeled sequents can typically be viewed as \emph{graphs} of Gentzen-style sequents, that is, graphs of (pairs of) multisets of logical formulae. This idea of defining inference rules over graphs of sequents has proven fruitful for the attainment and definition of sequent-style systems for diverse and sizable classes of logics~\cite{Gab96,Neg05,Sim94,Vig00}. Even more, sequent systems built within the labeled formalism exhibit a high degree of modularity, typically possess favorable proof-theoretic properties (e.g. invertibility of rules and cut-elimination), and algorithms exist that build labeled calculi for certain classes of logics on demand~\cite{CiaMafSpe13}. In spite of these desirable characteristics, the labeled formalism has certain drawbacks; in particular, labeled proof systems usually involve an unnecessarily complicated syntax (which incorporates redundant structures), include superfluous inference rules, and it is often cumbersome to establish the termination of associated proof search algorithms.
 
 By contrast, the nested sequent formalism only employs \emph{trees} of Gentzen-style sequents. Having originally been introduced by Bull~\cite{Bul92} and Kashima~\cite{Kas94}, and further expanded upon by Br\"unnler~\cite{Bru09} and Poggiolesi~\cite{Pog09}, the formalism of nested sequents avoids the pitfalls of labeled sequents, often yielding proof calculi that require less bureaucracy, are more compact, and where termination of proof search is more easily obtained. Moreover, nested sequent systems have a range of applications, 
 being used in knowledge integration algorithms~\cite{LyoAlv22}, serving as a basis for constructive interpolation and decidability techniques~\cite{FitKuz15,LyoTiuGorClo20,Lyo21thesis,TiuIanGor12}, and even being used to solve open questions about axiomatizability~\cite{IshKik07}. 
 
 More recently, it was discovered that nested sequent calculi could be extracted from labeled sequent calculi with the properties of the former transferred to the latter~\cite{CiaLyoRamTiu21,GorRam12,Pim18}. This discovery led to the formulation of a methodology referred to as \emph{structural refinement}~\cite{Lyo21a,Lyo21thesis,LyoBer19,Lyo21b}, which consists of two parts: first, one leverages general results regarding the (automated) construction of labeled sequent calculi (in possession of favorable properties) for a class of logics based on their semantics. Second, one transforms the labeled calculi into nested calculi via the introduction of \emph{propagation} or \emph{reachability rules} (cf.~\cite{CasCerGasHer97,Fit72,Lyo21thesis}) and the elimination of structural rules, followed by a notational translation. This entire process yields a method of transforming the semantics of a logic into a (cut-free) nested sequent calculus for the logic, being used to provide new and previously unknown nested sequent calculi~\cite{Lyo21b}. Although this method has been investigated and put to use for a variety of different \emph{propositional} modal logics, it has yet to be applied to \emph{first-order} modal logics. Thus, the first contribution of this paper is the application and study of structural refinement in the setting of \emph{first-order} modal logics, being used to obtain (to the best of the author's knowledge) the first (cut-free) nested sequent systems for such logics. Interestingly, this method `naturally' produces a generalized version of nested sequents, which include collections of terms (called \emph{signatures}) in their syntax, used to process quantificational data.
 
 

 As in the propositional setting, first-order modal logics admit a Kripke-style semantics, but unlike in the propositional setting, each world is associated with a domain of elements used to interpret quantificational formulae~\cite{Kri63}. One may then characterize distinct first-order modal logics by either imposing frame conditions on the accessibility relation of the model (as with propositional modal logics), or by imposing conditions on the domains associated with worlds~\cite{BraGhi07,FitMen98}. Regarding the latter, one obtains distinct first-order modal logics depending on if domains are permitted to be empty, if domains are permitted to increase or decrease along the accessibility relation, or if domains are required to be constant. In this paper, we consider a class of first-order modal logics that consists of extensions of a base logic $\qk$~\cite{BraGhi07} with seriality, general path conditions~\cite{CiaLyoRamTiu21}, and the aforementioned domain conditions. 
 
 In order to give this class of logics a uniform nested sequent presentation, we make use of so-called \emph{reachability rules}
, which possess two kinds of functionality: such rules may (i) propagate formulae along certain paths, and/or (ii) search for data along certain paths, within a labeled or nested sequent. The first kind of functionality is exhibited by the well-known class of \emph{propagation rules}~\cite{CasCerGasHer97,Fit72,GorPosTiu11}, which have proven vital for the provision of nested sequent systems for propositional modal and constructive logics~\cite{CiaLyoRamTiu21,GorPosTiu08,GorPosTiu11,LyoBer19}. Rules exhibiting the second kind of functionality 
 have only been defined more recently~\cite{Fit14,Lyo21a,Lyo21thesis}. Thus, reachability rules subsume the class of propagation rules, allowing for either or both kinds of functionality~\cite{Lyo21thesis}. 
 
 
 Our second contribution in this paper is the formulation of reachability rules for managing the introduction of quantificational formulae in labeled and nested sequent proofs. Motivated by the work in~\cite{GorPosTiu11,LyoBer19}, we parameterize such rules with (restricted versions of) semi-Thue systems~\cite{Pos47}, which encode along what paths formulae may be propagated and along what paths a rule should search for data. Such rules endow our nested sequent systems with a high degree of modularity as one nested calculus for a first-order modal logic may be straightforwardly converted into a nested calculus for another first-order modal logic by simply parameterizing the reachability rules with alternative semi-Thue systems.


The paper is organized as follows: In \sect~\ref{sec:log-prelims-I}, we define the considered class of first-order modal logics, as well as provide the grammar theoretical foundations for the formulation of our reachability rules. In \sect~\ref{sec:labeled-calc}, we present labeled sequent systems for first-order modal logics based on those of~\cite[\cptr~12.1]{NegPla11} and \cite[\cptr~6]{Vig00}. In the subsequent section (\sect~\ref{sec:struc-refinement}), we show how to define reachability rules in the context of labeled sequents, and apply the structural refinement methodology to derive nested sequent calculi from the given labeled calculi. In \sect~\ref{sec:nested-calculi}, we present the first (cut-free) nested sequent calculi for first-order modal logics, and in the final section (\sect~\ref{sec:conclusion}) we conclude.

%% file: prelims.tex
We let $\var := \{x, y, z, \ldots\}$ be a denumerable set of \emph{variables} and $\pred := \{p, q, r, \ldots\}$ be a set of denumerably many \emph{predicates} of each arity $n \in \mathbb{N}$ with \emph{propositional variables} predicates of arity zero. \emph{Atomic formulae} are of form $p(x_{1}, \ldots, x_{n})$, and we will often write them as $p(\vec{x})$. Similarly, we often write a list of variables $x_{1}, \ldots, x_{n}$ as $\vec{x}$.  The \emph{first-order modal language} $\lang$ is the set of formulae generated via the following grammar in BNF:
$$
\phi ::= p(\vec{x}) \ | \ \neg \phi \ | \ \phi \lor \phi \ | \ \Diamond \phi \ | \ \exists x \phi
$$
where $p$ ranges over $\pred$, and the variables $\vec{x} = x_{1}, \ldots, x_{n}$ and $x$ range over the set $\var$. We refer to a formula in $\lang$ as an \emph{$\lang$-formula}, and define $\phi \imp \psi := \neg \phi \lor \psi$, $\Box \phi := \neg \dia \neg \phi$, and $\forall x \phi := \neg \exists x \neg \phi$. As usual, we say that the occurrence of a variable $x$ in $\phi$ is \emph{free} given that $x$ does not occur within the scope of a quantifier. Also, we let $\phi[y/x]$ denote the substitution of the variable $y$ for each free occurrence of the variable $x$ in $\phi$. 

\begin{definition}[Frame, Model]
A \emph{frame} is a tuple $F = (W,R,D)$ such that:
\begin{itemize}

\item $W$ is a non-empty set $\{w, u, v, \ldots\}$ of \emph{worlds};

\item $R \ \subseteq W \times W$ is a binary relation on $W$;

\item $D$ is a \emph{domain function} mapping a world $w \in W$ to a (potentially empty) set $D(w)$ (called the \emph{domain of $w$}), and where $D(W) = \bigcup_{w \in W} D(w)$ is a non-empty 
 set disjoint from $W$ (called the \emph{domain of $W$}).
\end{itemize}
We say that a frame $F$ is of
\begin{itemize}

\item \emph{increasing domains} \iffi if $wRu$, then $D(w) \subseteq D(u)$;

\item \emph{decreasing domains} \iffi if $wRu$, then $D(u) \subseteq D(w)$;

\item \emph{constant domains} \iffi if $wRu$, then $D(w) = D(u)$;

\item \emph{non-empty domains} \iffi for each $w \in W$, $D(w) \neq \emptyset$.

\end{itemize}
We define a \emph{model} $M$ based on a frame $F$ to be an ordered pair $(F,V)$ where $F$ is a frame and $V$ is a \emph{valuation function} such that $V(p,w) \subseteq D(w)^{n}$ with $n \in \mathbb{N}$.\footnote{We note that $D(w)^{n}$ is the $n$-ary Cartesian product  $\underbrace{D(w) \times \cdots \times D(w)}_{n}$.} We make the simplifying assumption that for each world $w \in W$, $D(w)^{0} = \{()\}$ with $()$ the empty tuple, meaning $V(p,w) = \{()\}$ or $V(p,w) = \emptyset$, for any propositional variable $p$.

Given a model $M = (W,R,D,V)$, we define an \emph{assignment} over $M$ to be a function $\asmt : \var \to D(W)$ mapping variables to elements of the domain of $W$, and we let $\asmt[d/x]$ be the same as $\asmt$, but where the variable $x$ is mapped to $d \in D(W)$.
\end{definition}



\begin{definition}[Semantic Clauses]
Let $M = (W,R,D,V)$ be a model with $w \in W$ and let $\asmt$ be an assignment over $M$. The \emph{satisfaction relation} $\Vdash$ is defined as follows:
\begin{itemize}

\item $M,w,\asmt \Vdash p(x_{1}, \ldots, x_{n})$ \ifandonlyif $(\asmt(x_{1}), \ldots, \asmt(x_{n})) \in V(p,w)$;

\item $M,w,\asmt \Vdash \neg \phi$ \iffi $M,w,\asmt \not\Vdash \phi$;

\item $M,w,\asmt \Vdash \phi \lor \psi$ \ifandonlyif $M,w,\asmt \Vdash \phi$ or $M,w,\asmt \Vdash \psi$;



\item $M,w,\asmt \Vdash \Diamond \phi$ \ifandonlyif there exists a $u \in W$ such that $wRu$ and $M,u, \asmt \Vdash \phi$;


\item $M,w,\asmt \Vdash \exists x \phi$ \ifandonlyif there exists a $d \in D(w)$ such that $M,w,\asmt[d/x] \Vdash \phi$;


\item $M,w \Vdash \phi$ \iffi for all assignments $\asmt$ over $M$, $M,w,\asmt \Vdash \phi$;

\item $M \Vdash \phi$ \iffi for worlds $w \in W$, $M,w \Vdash \phi$.

\end{itemize}
An $\lang$-formula $\phi$ is \emph{valid} relative to a class of frames $\mathcal{F}$ \iffi $M \Vdash \phi$ for all models $M$ based on a frame $F \in \mathcal{F}$. 
\end{definition}

We define the base logic $\qk$ to be the set of all valid $\lang$-formulae over the class of all frames. We consider extensions of $\qk$ by imposing a set $\fcset$ of frame conditions (shown in \fig~\ref{fig:axioms-related-conditions} with their corresponding axioms) on the class of all frames. Following the top-down order of \fig~\ref{fig:axioms-related-conditions}, such frame conditions include: (1) the seriality condition $\serc$, (2) generalized path conditions $\gpc(n,k) := \forall w, u, v (w R^{n} u \land w R^{k} v \imp u R v)$ with $n,k \in \mathbb{N}$, (3) the increasing domain condition $\idc$ (with the \emph{Barcan formula} as the corresponding axiom), (4) the decreasing domain condition $\ddc$ (with the \emph{converse Barcan formula} as the corresponding axiom), (5) the constant domain condition $\cdc$, and (6) the non-empty domain condition $\nedc$.\footnote{For an introduction to first-order modal logics, frame conditions, and corresponding axioms, see Fitting and Mendelsohn~\cite{FitMen98}.} We use $\gpc(n,k)$ to denote a generalized path condition, use $\gpset$ to denote a set of such conditions, and note that such conditions cover well-known properties such as reflexivity (when $n = k = 0$), transitivity (when $k = 2$ and $n= 0$), symmetry (when $n = 1$ and $k = 0$), and Euclideanity (when $n = k = 1$). For a set $\fcset$ of frame conditions, we define $\mathcal{F}_{\fcset}$ to be the class of frames satisfying $\fcset$, and we define the logic $\ql$ to be the set of all valid $\lang$-formulae over the class $\mathcal{F}_{\fcset}$ (with $\mathsf{QK}(\emptyset) = \qk$). We make the simplifying assumption that if $\cdc \in \fcset$, then $\idc, \ddc \not\in \fcset$ as the former property implies the latter two. 

\begin{figure}[t]
\begin{center}
\bgroup
\def\arraystretch{1.1}
\begin{tabular}{| c | l | l |}
\hline
\qquad \qquad & Frame Condition & Corresponding Axiom\\
\hline
(1) & $\serc := \forall w \exists u (w R u) \quad$ & $\Box \phi \imp \dia \phi$\\
(2) & $\gpc(n,k) := \forall w, u, v (w R^{n} u \land w R^{k} v \imp u R v) \ $ & $\dia^{k} \phi \imp \Box^{n} \dia \phi$\\
(3) & $\idc := \forall w, u (wRu \imp D(w) \subseteq D(u))$ & $\Box \forall x \phi \imp \forall x \Box \phi \ $\\
(4) & $\ddc := \forall w, u (wRu \imp D(u) \subseteq D(w))$ & $\forall x \Box \phi \imp \Box \forall x \phi$\\
(5) & $\cdc := \forall w, u (wRu \imp D(u) = D(w))$ & $\forall x \phi \imp \phi[y/x]$\\
(6) & $\nedc := \forall w \exists x (x \in D(w))$ & $\forall x \phi \imp \exists x \phi$\\
\hline
\end{tabular}
\egroup
\end{center}

\caption{Axioms and their related frame conditions. We note that when $n=0$, the second frame condition is $\forall w, v (w R^{k} v \imp w R v)$, when $k = 0$, the frame condition is $\forall w, u (w R^{n} u \imp u R w)$, and when $n = k = 0$, the frame condition is $\forall w (w R w)$.}
\label{fig:axioms-related-conditions}
\end{figure}

%% file: grammar-prelims.tex
As will be seen 
 later on, a significant aspect of our nested calculi is the incorporation of inference (viz. reachability) rules whose applicability depends upon certain strings generated by a formal grammar. Therefore, the current section introduces the grammar-theoretic notions required to properly define such rules.

We define our \emph{alphabet} $\albet := \{\fd,\bd\}$ to be the set of \emph{characters} $\fd$ and $\bd$, and we let $\ques \in \albet$. We make use of the symbols $\fd$ and $\bd$ to encode information in inference rules about worlds in the \emph{future} and \emph{past} (resp.) of an accessibility relation $R$ in a model. We note that these symbols were adapted from the work in~\cite{GorPosTiu11}, and their usage mirrors their analogous meaning in the context of tense logics~\cite{GorPosTiu11,Kas94}. Additionally, we say that $\fd$ and $\bd$ are \emph{converses}, and define the \emph{converse operation} $\conv{\cdot}$ (on characters) as $\conv{\fd} := \bd$ and $\conv{\bd} := \fd$.

We let $\concat$ denote the typical \emph{concatenation operation} with $\varepsilon$ the \emph{empty string}. The set $\albet^{*}$ of \emph{strings over $\albet$} is defined to be the smallest set satisfying the following conditions: (i) $\albet \cup \{\varepsilon\} \subseteq \albet^{*}$, and (ii) $\text{If } \stra \in \albet^{*} \text{ and } \ques \in \albet \text{, then } \stra \concat \ques \in \albet^{*}$. We use $\stra$, $\strb$, $\strc$, \etc (potentially annotated) to denote strings from $\albetstr$. Also, we will often write the concatenation of two strings $s$ and $t$ as $\stra \cate \strb$ as opposed to $\stra \concat \strb$, and for the empty string $\empstr$, we have $\stra \cate \empstr = \empstr \cate \stra = \stra$. The converse operation on strings (adapted from~\cite{TiuIanGor12}) is defined on strings accordingly: (i) $\conv{\varepsilon} := \varepsilon$, and (ii) $\text{If } \stra = \ques_{1} \cdots \ques_{n} \text{, then } \conv{\stra} := \conv{\ques}_{n} \cdots \conv{\ques}_{1}$.

We now define \emph{$\albet$-systems}, which are special kinds of \emph{semi-Thue systems}~\cite{Pos47} that encode information about the conditions imposed on frames and models, and which will parameterize inference rules in our proof systems later on.

\begin{definition}[$\albet$-System]\label{def:grammar} We define a \emph{$\albet$-system} to be a set $\g$ of \emph{production rules} of the form $\ques \pto \stra$, where $\ques \in \albet$ and $\stra \in \albetstr$.
\end{definition}

Specific $\albet$-systems will be of use to us in this paper. In particular, for a set $\gpset$ of generalized path conditions, we define the $\albet$-system $\g(\gpset)$ as follows: $(\dia \pto \bd^{n} \cate \dia^{k}), (\bd \pto \bd^{k} \cate \fd^{n}) \in \g(\gpset)$ \iffi $\gpc(n,k) \in \gpset$. 
 We also define the $\albet$-systems $\sfour := \{\fd \pto \empstr, \bd \pto \empstr, \fd \pto \fd \fd, \bd \pto \bd \bd\}$ and $\sfive := \{\fd \pto \empstr, \bd \pto \empstr, \fd \pto \bd \fd, \bd \pto \bd \fd\}$, and describe their significance shortly; first, we define the notion of a \emph{derivation} and a \emph{language} in the context of $\albet$-systems.



\begin{definition}[Derivation, Language]\label{def:semi-thue-deriv-lang}
 Let $\g$ be a $\albet$-system. We write $\stra \pto \strb$ and say that  the string $\strb$ may be derived from the string $\stra$ in $\albetstr$ in \emph{one-step} \iffi there are strings $\stra', \strb' \in \albetstr$ and $\ques \pto \strc \in \thuesys$ such that $\stra = \stra' \cate \ques \cate \strb'$ and $\strb = \stra' \cate \strc \cate \strb'$. We define the \emph{derivation relation} $\dr$ to be the reflexive and transitive closure of $\pto$. For $\stra, \strb \in \albetstr$, we call $\stra \dr \strb$ a \emph{derivation of $\strb$ from $\stra$}, and define the \emph{length} of a derivation to be the minimal number of one-step derivations required to derive $\strb$ from $\stra$ in $\g$. Last, we define the \emph{language} $\glang(\stra) := \{\strb \ | \ \stra \dr \strb \}$, where $\stra \in \albet^{*}$.
\end{definition}
 
 As noted in~\cite[p.~143]{Lyo21thesis}, $\sfour$ encodes a notion of \emph{directed} reachability and $\sfive$ encodes a notion of \emph{undirected} reachability. Given a frame, recognizing what worlds are reachable from one another along paths of the accessibility relation $R$ is relevant in our context as frame conditions such as $\gpc(n,k)$ connect initial and terminal worlds, and $\idc$ and $\ddc$ require domains to contain certain elements along $R$-paths. To see how $\sfour$ and $\sfive$ encode notions of (un)directed reachability, let us think of $\fd$ as representing a forward move along the accessibility relation of a model, and $\bd$ as representing a backward move along the accessibility relation. Then, $\sfour$ defines the languages $\langsfour(\fd) = \{\fd^{n} \ | \ n \in \mathbb{N}\}$ and $\langsfour(\bd) = \{\bd^{n} \ | \ n \in \mathbb{N}\}$, where each string $\fd^{n}$ and $\bd^{n}$ (denoting a sequence of $n$ $\fd$'s and $n$ $\bd$'s, respectively, with $\fd^{0} = \bd^{0} = \empstr$) `connects' a world $w$ and $u$ if $u$ can be reached in $n$ forward moves or backward moves (resp.) along the accessibility relation from $w$. Similarly, $\sfive$ defines the language $\langsfive(\fd) = \langsfive(\bd) = \albetstr$ where each string $\ques_{1} \cdots \ques_{n}$ `connects' a world $w$ to $u$ if $u$ can be reached in $n$ forward and/or backward moves along the accessibility relation from $w$.


%% file: labelled-calc.tex

We present labeled sequent systems within our logical signature, based on those of~\cite[\cptr~6]{Vig00} and~\cite[\cptr~12.1]{NegPla11}, for first-order modal logics. These systems utilize a denumerable set $\lab := \{w,u,v,\ldots\}$ of \emph{labels} (which are occasionally annotated) and are comprised of inference rules that encode the semantics of first-order modal logics. Moreover, such inference rules operate on \emph{labeled sequents}, that is, formulae of the form $\Lambda := \rel, \Gamma \sar \Delta$, where $\rel$ is a multiset containing both \emph{relational atoms} of the form $wRu$ and \emph{domain atoms} of the form $x \in D(w)$, and $\Gamma$ and $\Delta$ are multisets of \emph{labeled formulae} of the form $w : \phi$, where $x$ ranges over $\var$, $w$ and $u$ range over $\lab$, and $\phi$ ranges over $\lang$. Given a multiset $\rel$ of relational atoms, we let $\lab(\rel)$ be the set all labels occurring in $\rel$.  

In addition, we define the \emph{composition} of two labeled sequents as
$$
(\rel_{1}, \Gamma_{1} \sar \Delta_{1}) \seqcomp (\rel_{2}, \Gamma_{2} \sar \Delta_{2}) := \rel_{1}, \rel_{2}, \Gamma_{1}, \Gamma_{2} \sar \Delta_{1}, \Delta_{2}
$$
 and define: (i) $\vec{x} \in D(w) := x_{1} \in D(w), \ldots, x_{n} \in D(w)$ for $\vec{x} = x_{1}, \ldots,x_{n}$, 
 (ii) $\Gamma \restriction w := \{\phi \ | \ w : \phi \in \Gamma\}$, and (iii) $w : \Gamma := \{w : \phi \ | \ \phi \in \Gamma \subseteq \lang\}$.\footnote{We use $\Gamma$ and $\Delta$ to represent (1) multisets of labeled formulae in the context of labeled sequents and (2) multisets of $\lang$-formulae in the context of nested sequents, and thus, the context determines the usage.} Last, to increase notational concision, we define a sequence of relational atoms $wR^{n}u := wRw_{1}, w_{1}Rw_{2}, \ldots, w_{n-1}Ru$ for some labels $w_{i}$ with $1 \leq i \leq n$ and note that $wR^{0}u := (w = u)$. 

\begin{figure}[t]

\begin{center}
\begin{tabular}{c c c} 
\AxiomC{}
\RightLabel{$\id$}
\UnaryInfC{$\rel, \Gamma, w : p(\vec{x}) \sar w : p(\vec{x}), \Delta$}
\DisplayProof

&

\AxiomC{$\rel, \Gamma \sar w : \phi, \Delta$}
\RightLabel{$\negl$}
\UnaryInfC{$\rel, \Gamma, w : \neg \phi \sar \Delta$}
\DisplayProof

&

\AxiomC{$\rel, \Gamma \sar w : \phi, w : \psi, \Delta$}
\RightLabel{$\disr$}
\UnaryInfC{$\rel, \Gamma \sar w : \phi \lor \psi, \Delta$}
\DisplayProof
\end{tabular}
\end{center}

\begin{center}
\begin{tabular}{c c}
\AxiomC{$\rel, \vec{x} \in D(w), \Gamma, w : p(\vec{x}) \sar \Delta$}
\RightLabel{$\domr$}
\UnaryInfC{$\rel, \Gamma, w : p(\vec{x}) \sar \Delta$}
\DisplayProof

&

\AxiomC{$\rel, \Gamma, w : \phi \sar \Delta$}
\AxiomC{$\rel, \Gamma, w : \psi \sar \Delta$}
\RightLabel{$\disl$}
\BinaryInfC{$\rel, \Gamma, w : \phi \vee \psi \sar \Delta$}
\DisplayProof
\end{tabular}
\end{center}

\begin{center}
\begin{tabular}{c c c}
\AxiomC{$\rel, w R u, \Gamma, u : \phi \sar \Delta$}
\RightLabel{$\dial^{\dag_{1}}$}
\UnaryInfC{$\rel, \Gamma, w : \dia \phi \sar \Delta$}
\DisplayProof

&

\AxiomC{$\rel, w R u, \Gamma \sar w : \dia \phi, u : \phi, \Delta$}
\RightLabel{$\diar$}
\UnaryInfC{$\rel, w R u, \Gamma \sar w : \dia \phi, \Delta$}
\DisplayProof

&

\AxiomC{$\rel, \Gamma, w : \phi \sar \Delta$}
\RightLabel{$\negr$}
\UnaryInfC{$\rel, \Gamma \sar w : \neg \phi, \Delta$}
\DisplayProof
\end{tabular}
\end{center}

\begin{center}
\begin{tabular}{c c}
\AxiomC{$\rel, y \in D(w), \Gamma, w : \phi[y/x] \sar \Delta$}
\RightLabel{$\existsl^{\dag_{2}}$}
\UnaryInfC{$\rel, \Gamma, w : \exists x \phi \sar \Delta$}
\DisplayProof

&

\AxiomC{$\rel, y \in D(w), \Gamma \sar w : \exists x \phi, w : \phi[y/x], \Delta$}
\RightLabel{$\existsr$}
\UnaryInfC{$\rel, y \in D(w), \Gamma \sar w : \exists x \phi, \Delta$}
\DisplayProof
\end{tabular}
\end{center}

\caption{The labeled calculus $\lqk$. The side condition $\dag_{1}$ ($\dag_{2}$) states that the rule is applicable only if $u$ ($y$, resp.) is \emph{fresh}, i.e. $u$ ($y$, resp.) does not occur in the conclusion of the rule.}
\label{fig:labeled-calculi}
\end{figure}

The labeled calculus $\lqk$ for the logic $\qk$ is given in \fig~\ref{fig:labeled-calculi}, and consists of the \emph{initial rule} $\ax$, pairs of left and right \emph{logical rules} for each logical connective, as well as the \emph{domain shift rule} $\domr$, which captures the semantic condition that 
 $V(p,w) \subseteq D(w)^{n}$ 
 in any model.\footnote{We note that the domain shift rule is missing from the labeled systems mentioned in~\cite[\cptr~6]{Vig00} and~\cite[\cptr~12.1]{NegPla11}, and thus, it is not clear if such systems are complete.} To obtain labeled systems for extensions of $\qk$ with a set $\fcset$ of frame conditions, i.e. for a logic $\ql$, we extend $\lqk$ with \emph{relational rules}, presented in \fig~\ref{fig:relational-rules}, and which are types of \emph{geometric structural rules}~\cite[p.~126]{Sim94}. 
 We refer to formulae that are explicitly presented in the premise(s) and conclusion of a rule as \emph{active}; e.g. $y \in D(w)$, $w : \phi[y/x]$, and $w : \exists x \phi$ are active in $\existsl$. Also, the \emph{height} of a proof is defined in the usual way as the longest sequence of sequents from the conclusion of a proof to an initial sequent (cf.~\cite{NegPla11}). 
 
 When extending $\lqk$ with relational rules, we refer to the resulting calculus as $\lql$, which is obtained in the following manner: (1) $\D \in \lql$ \iffi $\serc \in \fcset$, (2) $\sa \in \lql$ \iffi $\gpc(n,k) \in \fcset$, (3) $\idom \in \lql$ \iffi $\idc \in \fcset$, (4) $\ddom \in \lql$ \iffi $\ddc \in \fcset$, (5) $\cdom \in \lql$ \iffi $\cdc \in \fcset$, and (6) $\ned \in \lql$ \iffi $\nedc \in \fcset$. We make the simplifying assumption that if $\cdom \in \lql$, then $\idom, \ddom \not\in \lql$ since the latter two rules are redundant in the presence of the former, and we note that $\mathsf{G3QK}(\emptyset) = \lqk$. Furthermore, when $n=0$ or $k=0$ in $\sa$, i.e. when $\gpc(0,k) \in \fcset$, $\gpc(n,0) \in \fcset$, or $\gpc(0,0) \in \fcset$, the relational rules $(g_{0,k})$, $(g_{n,0})$, and $(g_{0,0})$ are respectively defined as:


\begin{center}
\begin{tabular}{c  @{\hskip 1em} c @{\hskip 1em} c}
\AxiomC{$\rel, w R^{k} v, w R v,  \Gamma \sar \Delta$}
\UnaryInfC{$\rel, w R^{k} v, \Gamma \sar \Delta$}
\DisplayProof

&

\AxiomC{$\rel, w R^{n} u, u R w,  \Gamma \sar \Delta$}
\UnaryInfC{$\rel, w R^{n} u, \Gamma \sar \Delta$}
\DisplayProof

&

\AxiomC{$\rel, w R w,  \Gamma \sar \Delta$}
\UnaryInfC{$\rel, \Gamma \sar \Delta$}
\DisplayProof
\end{tabular}
\end{center}

\begin{theorem}[$\lql$ Soundness and Completeness]\label{thm:lika-sound-complete} 
 $ \vdash w : \phi$ is derivable in $\lql$ \ifandonlyif $\phi \in \ql$.
\end{theorem}

\begin{proof}
 Similar to the proofs of 
 \thm~12.13 and 
 \thm~12.14 in~\cite[\cptr~12.1]{NegPla11}.
\end{proof}

\begin{figure}[t]

\begin{center}
\begin{tabular}{c c}
\AxiomC{$\rel, wRu, \Gamma \sar \Delta$}
\RightLabel{$\D^{\dag_{1}}$}
\UnaryInfC{$\rel, \Gamma \sar \Delta$}
\DisplayProof

&

\AxiomC{$\rel, w R^{n} u, w R^{k} v, u R v,  \Gamma \sar \Delta$}
\RightLabel{$\sa$}
\UnaryInfC{$\rel, w R^{n} u, w R^{k} v, \Gamma \sar \Delta$}
\DisplayProof
\end{tabular}
\end{center}

\begin{center}
\begin{tabular}{c c}
\AxiomC{$\rel, w R u, x \in D(w), x \in D(u), \Gamma \sar \Delta$}
\RightLabel{$\idom$}
\UnaryInfC{$\rel, w R u, x \in D(w), \Gamma \sar \Delta$}
\DisplayProof

&

\AxiomC{$\rel, w R u, x \in D(w), x \in D(u), \Gamma \sar \Delta$}
\RightLabel{$\ddom$}
\UnaryInfC{$\rel, w R u, x \in D(u), \Gamma \sar \Delta$}
\DisplayProof
\end{tabular}
\end{center}

\begin{center}
\begin{tabular}{c c}
\AxiomC{$\rel, x \in D(w), \Gamma \sar \Delta$}
\RightLabel{$\cdom$}
\UnaryInfC{$\rel, \Gamma \sar \Delta$}
\DisplayProof

&

\AxiomC{$\rel, x \in D(w), \Gamma \sar \Delta$}
\RightLabel{$\ned^{\dag_{2}}$}
\UnaryInfC{$\rel, \Gamma \sar \Delta$}
\DisplayProof
\end{tabular}
\end{center}

\caption{Relational rules. The side condition $\dag_{1}$ states that $\D$ is applicable only if $u$ is fresh, and $\dag_{2}$ states that $\ned$ is applicable only if $x$ is fresh.}
\label{fig:relational-rules}
\end{figure}


%% file: struc-refinement.tex
 We now \emph{structurally refine} the labeled systems of the previous section, that is, we expand each calculus $\lql$ with a set of reachability rules (generalized versions of the $\diar$ and $\existsl$ rules) that permit the elimination of relational rules (viz. $\sa$, $\idom$, $\ddom$, $\cdom$, and $\ned$) from any given derivation. This process is an instance of a methodology introduced in~\cite{Lyo21thesis} 
 to derive nested systems from labeled systems, and begets `refined' labeled calculi (each dubbed $\qll$), which only require labeled sequents of a `treelike' structure in proofs.

The reachability rules we introduce 
 possess a rather complex functionality that will be detailed in the first half of this section with an example given below. In particular, our reachability rules view labeled sequents as automata, and enable formulae to be (bottom-up) introduced to a labeled sequent given that certain paths of relational atoms (corresponding to strings generated by a $\albet$-system) exist in a sequent. To make the functionality of such rules precise, we define \emph{propagation graphs} and \emph{propagation paths}---concepts initially defined in the context of display and nested calculi~\cite{GorPosTiu11}, and subsequently transported to the setting of labeled calculi~\cite{CiaLyoRamTiu21,LyoBer19}.

\begin{definition}[Propagation Graph]\label{def:propagation-graph} We define a \emph{propagation graph} of a labeled sequent $\Lambda := \rel, \Gamma \sar \Delta$ to be a pair $\prgr{\rel} := (V,E)$ such that (i) $(w,\vx) \in V$ \iffi $w \in \lab(\rel)$, $\vx = \{x_{1}, \ldots, x_{n}\} \subset \var$, and $x_{1} \in D(w), \ldots, x_{n} \in D(w)$ are all domain atoms in $\rel$ associated with the label $w$, and (ii) $(w,\fd,u),(u,\bd,w) \in E$ \iffi $wRu \in \rel$. We will often write $w \in \prgr{\rel}$ to mean $w \in \prgrdom$, and $(w,\ques,u) \in \prgr{\rel}$ to mean $(w,\ques,u) \in \prgredges$.
\end{definition}

\begin{definition}[Propagation Path~\cite{Lyo21b}]\label{def:propagation-path} Let $\Lambda := \rel, \Gamma \sar \Delta$ be a labeled sequent. We define a \emph{propagation path from $w_{1}$ to $w_{n}$ in $\prgr{\rel} := (V,E)$} to be a sequence of the following form:
$$
\ppath(w_{1},w_{n}) := w_{1}, \ques_{1}, w_{2}, \ques_{2}, w_{3}, \ldots , \ques_{n-1}, w_{n}
$$
such that $(w_{1}, \ques_{1}, w_{2}) , (w_{2}, \ques_{2}, w_{3}), \ldots, (w_{n-1}, \ques_{n-1}, w_{n}) \in \prgredges$. Given a propagation path as above, we define its \emph{converse} as shown below left and its \emph{string} as shown below right:
$$
\conv{\ppath}(w_{n},w_{1}) := w_{n}, \conv{\ques}_{n-1}, w_{n-1}, \conv{\ques}_{n-2}, \ldots, \conv{\ques}_{1}, w_{1}
\qquad
\stra_{\ppath}(w_{1},w_{n}) := \ques_{1} \cate \ques_{2} \cate \cdots \cate \ques_{n-1}
$$
 We let $\emppath(w,w) := w$ be an \emph{empty path} with its string defined as $\stra_{\emppath}(w,w) := \empstr$.
\end{definition}

 With the above notions, we can now define the operation of our reachability rules $\prdia$, $\prexi$, and $\prexii$ which are displayed in \fig~\ref{fig:reachability-rules}. 
 The side conditions dictating the operation of such rules are specified in \dfn~\ref{def:side-conditions} below. As mentioned previously, reachability rules incorporate two kinds of functionality: (i) the propagation of formulae along paths within a sequent, and (ii) the search of data along paths within a sequent. Rules that only incorporate the first kind of functionality occur in the literature, and constitute \emph{propagation rules}~\cite{CasCerGasHer97,Fit72,GorPosTiu11,Sim94}. In this regard, $\prdia$ serves as a propagation rule proper. On the other hand, the new rules $\prexi$ and $\prexii$ serve as reachability rules, exhibiting the second kind of functionality mentioned above. 

\begin{figure}[t]

\begin{center}
\begin{tabular}{c c}
\AxiomC{$\rel, \Gamma \sar w : \dia \phi, u : \phi, \Delta$}
\RightLabel{$\prdia^{\dag_{1}(\fcset)}$}
\UnaryInfC{$\rel, \Gamma \sar w : \dia \phi, \Delta$}
\DisplayProof

&

\AxiomC{$\rel, \Gamma \sar w : \phi[y/x], w : \exists x \phi, \Delta$}
\RightLabel{$\prexi^{\dag_{2}(\fcset)}$}
\UnaryInfC{$\rel, \Gamma \sar w : \exists x \phi, \Delta$}
\DisplayProof
\end{tabular}
\end{center}
\begin{center}
\begin{tabular}{c}
\AxiomC{$\rel, y \in D(u), \Gamma \sar w : \phi[y/x], w : \exists x \phi, \Delta$}
\RightLabel{$\prexii^{\dag_{3}(\fcset)}$}
\UnaryInfC{$\rel, \Gamma \sar w : \exists x \phi, \Delta$}
\DisplayProof
\end{tabular}
\end{center}

\caption{Labeled reachability rules. The side condition $\dag_{1}(\fcset)$ is defined in \dfn~\ref{def:side-conditions}, $\dag_{2}(\fcset)$ is defined in \fig~\ref{fig:prex-side-cond}, and $\dag_{3}(\fcset)$ is defined in \fig~\ref{fig:prex-2-side-cond}.}
\label{fig:reachability-rules}
\end{figure}

 Each rule $\prdia$, $\prexi$, and $\prexii$ relies on a side condition (see \dfn~\ref{def:side-conditions} below) that may check to see if a propagation path of a certain shape exists within the propagation graph of a labeled sequent. For a $\albet$-system $\g$ and a labeled sequent $\Lambda := \rel, \Gamma \sar \Delta$, such conditions typically involve a statement of the following form: \emph{There exists a propagation path $\ppath(w,u)$ from $w$ to $u$ in the propagation graph $\prgr{\rel}$ such that the string $\stra_{\ppath}(w,u)$ is in the language $\glang(\ques)$.} We express this statement in a more concise manner as shown below:
$$
\exists \ppath(w,u) \in \prgr{\rel} \big(\stra_{\ppath}(w,u) \in \glang(\ques)\big)
$$
We remark that $\prexi$ and $\prexii$ also rely on \emph{availability} notions (cf.~\cite{Fit14,Lyo21thesis}), which employ statements of the above shape.

\begin{definition}[$(\g,\ques)$-available]\label{def:S4-S5-available} Let $\Lambda := \rel, \Gamma \sar \Delta$ be a labeled sequent and $\g$ be a $\albet$-system. We define a variable $x$ to be \emph{$(\g,\ques)$-available} for a label $w$ in $\Lambda$ \iffi for some label $u$, $(u,\vx) \in \prgr{\rel}$ with $x \in \vx$ such that $\exists \ppath(w,u) \in \prgr{\rel} \big(\stra_{\ppath}(w,u) \in L_{\g}(\ques)\big)$.
\end{definition}

 Let us now formally specify the side conditions $\dag_{1}(\fcset)$ -- $\dag_{3}(\fcset)$ imposed on the $\prdia$, $\prexi$, and $\prexii$ rules, which depend upon the frame conditions present in $\fcset$.

\begin{definition}[Side Conditions]\label{def:side-conditions} Let $\gpset = \{\gpc(n,k) \ | \ \gpc(n,k) \in \fcset\}$. Then, for the $\prdia$ rule, the side condition is defined accordingly:
$$
\dag_{1}(\fcset) := \exists \ppath(w,u) \in \prgr{\rel} (\stra_{\ppath}(w,u) \in L_{\thuesys(\gpset)}(\fd)).
$$
 For the $\prexi$ and $\prexii$ rules, the side condition depends on if $\idc$, $\ddc$, or $\cdc$ occur within $\fcset$ or not. The conditions for $\prexi$ and $\prexii$ are presented in \fig~\ref{fig:prex-side-cond} and \fig~\ref{fig:prex-2-side-cond}, respectively.
\end{definition}

\begin{figure}[t]
\begin{center}
\bgroup
\def\arraystretch{1.1}
\begin{tabular}{| l | l |}
\hline
 Property of $\fcset$ & Corresponding Side Condition for $\prexi$\\
\hline
$\idc, \ddc \in \fcset$, $\cdc \not\in \fcset$ & $\dag_{2}(\fcset) :=$  ``$y$ is $(\sfive,\fd)$-available for $w$.''\\
$\idc \in \fcset$, $\ddc, \cdc \not\in \fcset$ & $\dag_{2}(\fcset) :=$  ``$y$ is $(\sfour \cup \thuesys(\gpset),\bd)$-available for $w$.''\\
$\ddc \in \fcset$, $\idc, \cdc \not\in \fcset$ & $\dag_{2}(\fcset) :=$  ``$y$ is $(\sfour \cup \thuesys(\gpset),\fd)$-available for $w$.''\\
$\idc, \ddc, \cdc \not\in \fcset$ & $\dag_{2}(\fcset) :=$  ``$(y \in D(w)) \in \rel$.''\\
$\cdc \in \fcset$ & $\dag_{2}(\fcset) :=$  ``$y \in \var$.''\\
\hline
\end{tabular}
\egroup
\end{center}

\caption{The side condition $\dag_{2}(\fcset)$ for $\prexi$, which depends on the contents of $\fcset$. Note that there are only five cases above as we have made the simplifying assumption that if $\cdc \in \fcset$, then $\idc,\ddc \not\in \fcset$.}
\label{fig:prex-side-cond}
\end{figure}

\begin{example}\label{ex:propagation-graph-path} Let $\Lambda := wRv, wRu, y \in D(w), z \in D(u) \sar v : \exists x p(x), u : \dia (q \lor r)$ with $\rel := wRv, wRu$. A graphical depiction of $\prgr{\rel}$ is given below: 
\begin{center}
\begin{minipage}[t]{.5\textwidth}
\xymatrix{
  (v,\emptyset)\ar@/^-1pc/@{.>}[rr]|-{\bd} & &  (w,\{y\})\ar@/^1pc/@{.>}[rr]|-{\fd}\ar@/^-1pc/@{.>}[ll]|-{\dia} & &  (u,\{z\})\ar@/^1pc/@{.>}[ll]|-{\bd}
}
\end{minipage}
\end{center}
First, observe that if $\gpc(1,1) \in \fcset$ (i.e. Euclideanity is imposed as a frame condition), then $\fd \pto \bd \cate \fd \in \thuesys(\gpset)$, and since the propagation path $u, \bd, w, \fd, v$ exists in $\prgr{\rel}$ with $\bd \cate \fd \in L_{\thuesys(\gpset)}(\fd)$, it follows that $\prdia$ can be bottom-up applied to $\Lambda$, letting us derive the labeled sequent $wRv, wRu, y \in D(w), z \in D(u) \sar v : \exists x p(x), u : \dia (q \lor r), v : p \lor r$.

Second, if we assume that $\idc \in \fcset$, but $\ddc \not\in \fcset$, then $y$ is $(\sfour \cup \thuesys(\gpset),\bd)$-available for $v$ since there is a propagation path $v, \bd, w$ with $\bd \in L_{\sfour \cup \thuesys(\gpset)}(\bd)$, and $(w,\{y\}) \in \prgr{\rel}$. Hence, $wRv, wRu, y \in D(w), z \in D(u) \sar v : \exists x p(x), v : p(y), u : \dia (q \lor r)$ can be bottom-up derived by an application of $\prexi$.
\end{example}

\begin{remark}\label{rem:diar-existsr-reachability-instances}
The $\diar$ rule is an instance of $\prdia$. If $\cdc \not\in \fcset$, then $\existsr$ is an instance of $\prexi$, and if $\cdc \in \fcset$, then $\existsr$ corresponds to an application of $\prexi$ followed by $\cdom$. Also, notice that if $\idc, \ddc, \cdc \not\in \fcset$, then $\prexi$ is identical to the $\existsr$ rule and $\prexii$ corresponds to an application of $\existsr$ followed by an application of $\ned$ (given that $\nedc \in \fcset$).
\end{remark}

\begin{figure}[t]
\begin{center}
\bgroup
\def\arraystretch{1.1}
\begin{tabular}{| l | l |}
\hline
 Property of $\fcset$ & Corresponding Side Condition for $\prexii$\\
\hline
$\idc, \ddc \in \fcset$, $\cdc \not\in \fcset$ & $\dag_{3}(\fcset) :=$ ``$\exists \ppath(w,u) \in \prgr{\rel} (\stra_{\ppath}(w,u) \in L_{\sfive}(\fd))$ and $y$ is fresh.''\\
$\idc \in \fcset$, $\ddc, \cdc \not\in \fcset$ & $\dag_{3}(\fcset) :=$ ``$\exists \ppath(w,u) \in \prgr{\rel} (\stra_{\ppath}(w,u) \in L_{\sfour \cup \thuesys(\gpset)}(\bd))$ and $y$ is fresh.''\\
$\ddc \in \fcset$, $\idc, \cdc \not\in \fcset$ & $\dag_{3}(\fcset) :=$ ``$\exists \ppath(w,u) \in \prgr{\rel} (\stra_{\ppath}(w,u) \in L_{\sfour \cup \thuesys(\gpset)}(\fd))$ and $y$ is fresh.''\\
$\idc, \ddc, \cdc \not\in \fcset$ & $\dag_{3}(\fcset) :=$  ``$w = u$ and $y$ is fresh.''\\
$\cdc \in \fcset$ & $\dag_{3}(\fcset) :=$  ``$y$ is fresh.''\\
\hline
\end{tabular}
\egroup
\end{center}

\caption{The side condition $\dag_{3}(\fcset)$ for $\prexii$, which depends on the contents of $\fcset$.}
\label{fig:prex-2-side-cond}
\end{figure}

\begin{definition}[Refined Labeled Calculus] Let $\gpset = \{\gpc(n,k) \ | \ \gpc(n,k) \in \fcset\}$. We define the \emph{refined labeled calculus} $\qll := (\lql \setminus \mathbf{R}) \cup \{\prdia,\prexi\} \cup \{\prexii \ | \ \nedc \in \fcset\}$, where
$$
\mathbf{R} = \{\sa \ | \ \gpc(n,k) \in \gpset\} \cup \{\idom, \ddom, \cdom, \ned\}.
$$
 In other words, $\qll$ extends $\lql$ with the $\prdia$ and $\prexi$ rules regardless, but only includes $\prexii$ \iffi $\nedc \in \fcset$, and excludes the $\sa$, $\idom$, $\ddom$, $\cdom$, and $\ned$ relational rules.
\end{definition}

We now argue that each refined labeled calculus $\qll$ is complete by means of a proof transformation procedure. In particular, we show that through the elimination of relational rules and the introduction of reachability rules, any proof in $\lql$ can be transformed into a proof in $\qll$. To provide intuition, we begin by showing how a concrete proof in $\lql$ can be transformed into a proof in $\qll$, and afterward, we prove in general that such a transformation can always be performed.

\begin{example}\label{ex:permutation} Let $\fcset := \{\idc,\gpc(0,2)\}$ with $\gpset = \{\gpc(0,2)\}$, i.e. $\fcset$ contains the increasing domain condition and transitivity as frame conditions. Therefore, $\prdia$ is parameterized by the $\albet$-system $\thuesys(\gpset) = \{\fd \pto \fd \fd, \bd \pto \bd \bd\}$, $\prexi$ is parameterized by the $\albet$-system $\sfour \cup \thuesys(\gpset)$, and $\prexii$ does not occur in $\qll$ as $\nedc \not\in \fcset$. Moreover, suppose we are given the following derivation in $\lql$. We will show how the relational rules $\idom$ and $(g_{0,2})$ can be permuted upward until they are eliminated from a given derivation, yielding a proof in $\qll$.

\begin{center}
\AxiomC{}
\RightLabel{$\id$}
\UnaryInfC{$wRu, uRv, wRv, y \in D(w), y \in D(v), v : p(y) \sar v : \exists x p(x), v : p(y)$}
\RightLabel{$\existsr$}
\UnaryInfC{$wRu, uRv, wRv, y \in D(w), y \in D(v), v : p(y) \sar v : \exists x p(x)$}
\RightLabel{$\idom$}
\UnaryInfC{$wRu, uRv, wRv, y \in D(w), v : p(y) \sar v : \exists x p(x)$}
\RightLabel{$(g_{0,2})$}
\UnaryInfC{$wRu, uRv, y \in D(w), v : p(y) \sar v : \exists x p(x)$}
\DisplayProof
\end{center}

First, as shown in the derivation below, $\idom$ can be permuted above $\existsr$. Observe that if $\idom$ is applied to the initial sequent, then a propagation path $v, \bd, w$ exists in $\prgr{\rel}$, where $\rel = wRu, uRv, wRv$, such that $\bd \in L_{\sfour \cup \thuesys(\gpset)}(\bd)$, showing that $\prexi$ can indeed be applied after $\idom$.

\begin{center}
\AxiomC{}
\RightLabel{$\id$}
\UnaryInfC{$wRu, uRv, wRv, y \in D(w), y \in D(v), v : p(y) \sar v : \exists x p(x), v : p(y)$}
\RightLabel{$\idom$}
\UnaryInfC{$wRu, uRv, wRv, y \in D(w), v : p(y) \sar v : \exists x p(x), v : p(y)$}
\RightLabel{$\prexi$}
\UnaryInfC{$wRu, uRv, wRv, y \in D(w), v : p(y) \sar v : \exists x p(x)$}
\RightLabel{$(g_{0,2})$}
\UnaryInfC{$wRu, uRv, y \in D(w), v : p(y) \sar v : \exists x p(x)$}
\DisplayProof
\end{center}

We now observe that the conclusion of $\idom$ in the proof above is an instance of $\id$, and hence, the application of $\idom$ may be eliminated and replaced by an instance of $\id$ as shown below. Furthermore, as is also shown below, $(g_{0,2})$ may be permuted above $\prexi$, since if $(g_{0,2})$ is applied to the conclusion of $\id$, then there exists a propagation path $v, \bd, u, \bd, w$ in $\prgr{\rel'}$, where $\rel' = wRu, uRv$, such that $\bd \bd \in L_{\sfour \cup \thuesys(\gpset)}(\bd)$. Hence, the side condition of $\prexi$ continues to hold after $(g_{0,2})$ is applied.

\begin{center}
\AxiomC{}
\RightLabel{$\id$}
\UnaryInfC{$wRu, uRv, wRv, y \in D(w), v : p(y) \sar v : \exists x p(x), v : p(y)$}
\RightLabel{$(g_{0,2})$}
\UnaryInfC{$wRu, uRv, y \in D(w), v : p(y) \sar v : \exists x p(x), v : p(y)$}
\RightLabel{$\prexi$}
\UnaryInfC{$wRu, uRv, y \in D(w), v : p(y) \sar v : \exists x p(x)$}
\DisplayProof
\end{center}

Last, we observe that the conclusion of $(g_{0,2})$ is an instance of $\id$, meaning $(g_{0,2})$ can be eliminated from the proof and replaced by $\id$, as shown below.

\begin{center}
\AxiomC{}
\RightLabel{$\id$}
\UnaryInfC{$wRu, uRv, y \in D(w), v : p(y) \sar v : \exists x p(x), v : p(y)$}
\RightLabel{$\prexi$}
\UnaryInfC{$wRu, uRv, y \in D(w), v : p(y) \sar v : \exists x p(x)$}
\DisplayProof
\end{center}

As seen in the above example, all relational rules have been eliminated from a $\lql$ derivation, yielding a derivation in $\qll$. Significantly, observe that the relational atoms $wRu, uRv$ occurring in the two labeled sequents of the proof form an $R$-path of length two, which is a type of tree. This is a characteristic feature of proofs in $\qll$ (proven in \lem~\ref{lem:labeled-tree-derivations} below), namely, that proofs only require labeled sequents of a treelike structure---a fact leveraged to derive nested systems for first-order modal logics in the next section.
\end{example}

\begin{lemma}\label{lem:permutation-1}
Let $\gpc(n,k) \in \fcset$ and $\ru \in \{\prdia, \prexi, \prexii\}$. Then, $\sa$ can be permuted above $\ru$, that is, any $\ru$ inference followed by a $\sa$ inference can be transformed into a $\sa$ inference followed by an $\ru$ inference.
\end{lemma}

\begin{proof} We prove the case where $\ru$ is $\prexi$ as the case where $\ru$ is $\prexii$ is shown similarly, and the case where $\ru$ is $\prdia$ follows from \cite[\lem~5.8]{CiaLyoRamTiu21}. Therefore, suppose we are given the following:
\begin{center}
\AxiomC{$\rel, y \in D(u'), w R^{n} u, w R^{k} v, u R v,  \Gamma \sar w' : \exists x \phi, w' : \phi[y/x], \Delta$}
\RightLabel{$\prexi$}
\UnaryInfC{$\rel, y \in D(u'), w R^{n} u, w R^{k} v, u R v,  \Gamma \sar w' : \exists x \phi, \Delta$}
\RightLabel{$\sa$}
\UnaryInfC{$\rel, y \in D(u'), w R^{n} u, w R^{k} v,  \Gamma \sar w' : \exists x \phi, \Delta$}
\DisplayProof
\end{center}
 We assume that the relational atom $u R v $ is active in the $\prexi$ inference, since the cases where it is not active are easily resolved as $\sa$ can be freely permuted above $\prexi$. Let $\rel_{1} := \rel,y \in D(u'), w R^{n} u, w R^{k} v, u R v$ and $\rel_{2} := \rel, y \in D(u'), w R^{n} u, w R^{k} v$. We note that since the relational atom $u R v$ is assumed active in $\prexi$, it follows that $u R v$ occurs along a propagation path $\ppath(u',w')$ such that $(u',\vx) \in \prgr{\rel_{1}}$ with $y \in \vx$ and $\exists \ppath(u',w') \in \prgr{\rel_{1}} \big(\stra_{\ppath}(u',w') \in L_{\g \cup \thuesys(\gpset)}(\ques)\big)$. We note that $\g \in \{\sfour, \sfive\}$ and $\ques \in \{\fd,\bd\}$, with both being determined on the basis of if $\idc$ and/or $\ddc$ occur in $\fcset$. We assume that $\ques = \fd$ for simplicity, noting that the proof is similar when $\ques = \bd$.

 We now show that $\exists \ppath'(u',w') \in \prgr{\rel_{2}} (\stra_{\ppath'}(u',w') \in L_{\thuesys \cup \thuesys(\gpset)}(\fd))$ to complete the proof. To do so, we build a propagation path $\ppath'(u',w')$ by replacing each occurrence of $u, \dia, v$ in $\ppath(u',w')$ by 
 $u, \bd, u_{1}, \ldots, u_{n-1}, \bd, w, \fd, w_{1}, \ldots, w_{k-1}, \fd, v$ (corresponding to the relational atoms $w R^{n} u, w R^{k} v \in \rel_{2}$), and replacing each occurrence of $v, \bd, u$ in $\ppath(u',w')$ by $ v, \bd, w_{k-1}, \ldots, w_{1}, \bd, w, \fd, u_{n-1}, \ldots, u_{1}, \fd, u$ (also corresponding to the relational atoms $w R^{n} u, w R^{k} v \in \rel_{2}$). We let $\ppath'(u',w')$ denote the propagation path obtained from $\ppath(u',w')$ in the just described way. As the only difference between $\prgr{\rel_{1}}$ and $\prgr{\rel_{2}}$ is that the former is guaranteed to contain the edges $(u,\fd,v)$ and $(v,\bd,u)$ corresponding to $u R v$, whereas the latter is not, we have that $\ppath'(u',w')$ is a propagation path in $\prgr{\rel_{2}}$ as $\ppath'(u',w')$ omits $u, \fd, v$ and $v,\bd, u$. 

 Last, we additionally need to show that $\stra_{\ppath'}(u',w') \in L_{\thuesys \cup \thuesys(\gpset)}(\fd)$. By assumption, $\stra_{\ppath}(u',w') \in L_{\thuesys \cup \thuesys(\gpset)}(\fd)$, implying that $\fd \pto_{\thuesys \cup \thuesys(\gpset)}^{*} \stra_{\ppath}(u',w')$ by \dfn~\ref{def:semi-thue-deriv-lang}. Also, due to the fact that $\sa$ is a rule in $\lql$, we know that $\fd \pto \bd^{n} \cate \fd^{k},\bd \pto \bd^{k} \cate \fd^{n} \in \thuesys \cup \thuesys(\gpset)$. Applying $\fd \pto \bd^{n} \cate \fd^{k}$ to each occurrence of $\fd$ in $\stra_{\ppath}(u',w')$ corresponding to the edge $(u,\fd,v)$, and applying $\bd \pto \bd^{k} \cate \fd^{n}$ to each occurrence of $\bd$ in $\stra_{\ppath}(u',w')$ corresponding to the edge $(v,\bd,u)$ (with both edges corresponding to the relational atom $u R v$), we obtain the string $\stra_{\ppath'}(u',w') \in L_{\thuesys \cup \thuesys(\gpset)}(\fd)$. Hence, we may permute $\sa$ above $\prexi$, giving us:
 \begin{center}
\AxiomC{$\rel, y \in D(z), w R^{n} u, w R^{k} v, u R v,  \Gamma \sar w' : \exists x \phi, w' : \phi[y/x], \Delta$}
\RightLabel{$\sa$}
\UnaryInfC{$\rel, y \in D(z), w R^{n} u, w R^{k} v,  \Gamma \sar w' : \exists x \phi, w' : \phi[y/x], \Delta$}
\RightLabel{$\prexi$}
\UnaryInfC{$\rel, y \in D(z), w R^{n} u, w R^{k} v,  \Gamma \sar w' : \exists x \phi, \Delta$}
\DisplayProof
\end{center}
\end{proof}


\begin{lemma}\label{lem:permutation-2} (1) If $\idc \in \fcset$, then $\idom$ can be permuted above $\prexi$, 
 and (2) If $\ddc \in \fcset$, then $\ddom$ can be permuted above $\prexi$, (3) If $\cdc \in \fcset$, then $\cdom$ can be permuted above $\prexi$.
\end{lemma}

\begin{proof} Let $\gpset = \{\gpc(n,k) \ | \ \gpc(n,k) \in \fcset\}$.   We prove statement (2) as statement (1) is similar and statement (3) is trivial. Suppose we are given the following:
\begin{center}
\AxiomC{$\rel, y \in D(u), y \in D(v), u R v, \Gamma \sar w : \exists x \phi, w : \phi[y/x], \Delta$}
\RightLabel{$\prexi$}
\UnaryInfC{$\rel, y \in D(u), y \in D(v), u R v, \Gamma \sar w : \exists x \phi, \Delta$}
\RightLabel{$\ddom$}
\UnaryInfC{$\rel, y \in D(v), u R v, \Gamma \sar w : \exists x \phi, \Delta$}
\DisplayProof
\end{center}
 We assume that the domain atom $y \in D(u)$ is active in $\prexi$ above, since the case where it is not active can be easily resolved by permuting $\ddom$ above $\prexi$. Also, we let $\rel_{1} := \rel, y \in D(u), y \in D(v), u R v$ and $\rel_{2} := \rel, y \in D(v), u R v$. By the side condition on $\prexi$, we know that there exists a propagation path $\ppath(w,u)$ such that $(u,\vx) \in \prgr{\rel_{1}}$ with $y \in \vx$ and $\exists \ppath(w,u) \in \prgr{\rel_{1}} \big(\stra_{\ppath}(w,u) \in L_{\sfour \cup \thuesys(\gpset)}(\fd)\big)$. We assume w.l.o.g. that $\ppath(w,u)$ is a minimal propagation path between $w$ and $u$, which either (i) traverses $v$ before it reaches $u$, or (ii) does not traverse $v$ before it reaches $u$. We prove each case in turn:
 
\textit{(i).} If $\ppath(w,u)$ traverses $v$ before it reaches $u$, then the propagation path is of the form $\ppath'(w,v),\bd,u$. We observe that $\ppath'(w,v)$ is a propagation path from $w$ to $v$ in $\rel_{2}$ and that $\stra_{\ppath}(w,u) = \stra_{\ppath'}(w,v) \concat \bd$. Since $\bd \pto \empstr \in \sfour$, it follows that $\stra_{\ppath'}(w,v) \in L_{\sfour \cup \thuesys(\gpset)}(\fd)$, which shows that $\ddom$ can be permuted above $\prexi$ as the side condition of $\prexi$ continues to hold after $\ddom$ is applied.

\textit{(ii).} Suppose $\ppath(w,u)$ does not traverse $v$ before it reaches $u$. If $\ppath(w,u) = \empstr$, then the case is easily resolved, so we assume that $\ppath(w,u)$ is either of the form $\ppath(w,z),\fd,u$ or $\ppath(w,z),\bd,u$. In the former case, since $\ppath'(w,v) = \ppath(w,z),\fd,u,\fd,v$ is a propagation path in $\prgr{\rel_{2}}$, $\stra_{\ppath}(w,u) \in L_{\sfour \cup \thuesys(\gpset)}(\fd)$, and $\fd \pto \fd \fd \in \sfour$, it follows that $\stra_{\ppath'}(w,v) = \stra_{\ppath}(w,u) \concat \fd = \stra_{\ppath}(w,z) \concat \fd \concat \fd \in L_{\sfour \cup \thuesys(\gpset)}(\fd)$, showing that $\ddom$ can be permuted above $\prexi$ in this case. Let us now consider the latter case. If $\ppath(w,u)$ is of the form $\ppath(w,z),\bd,u$ with $\stra_{\ppath}(w,z) \concat \bd \in L_{\sfour \cup \thuesys(\gpset)}(\fd)$, then since $L_{\sfour}(\fd) = \{\fd^{n} \ | \ n \in \mathbb{N}\}$, we know that $\stra_{\ppath}(w,z) \concat \bd \in L_{\sfour \cup \thuesys(\gpset)}(\fd)$ due to an application of a production rule $\fd \pto \bd^{n} \fd^{k}$ with $n \geq 1$ in a derivation $\fd \pto_{\sfour \cup \thuesys(\gpset)}^{*} \stra_{\ppath}(w,z) \concat \bd$. Because $\fd \pto \empstr, \bd \pto \empstr \in \sfour$, it follows that $\ques \pto_{\sfour \cup \thuesys(\gpset)}^{*} \ques^{-1}$ for $\ques \in \{\fd, \bd\}$, thus showing that $L_{\sfour \cup \thuesys(\gpset)}(\fd) = \albet^{*}$. Hence, since $\ppath'(w,v) = \ppath(w,z),\bd,u,\fd,v$ is a propagation path in $\prgr{\rel_{2}}$ and $L_{\sfour \cup \thuesys(\gpset)}(\fd) = \albet^{*}$, we know that $\stra_{\ppath}(w,z) \concat \bd \concat \fd \in L_{\sfour \cup \thuesys(\gpset)}(\fd)$, meaning $\ddom$ can be permuted above $\prexi$.
\end{proof}

\begin{theorem}\label{thm:lab-to-ref-lab}
Every derivation in $\lql$ can be transformed into one in $\qll$.
\end{theorem}

\begin{proof} Suppose we are given a proof in $\lql$, which is a proof in $\lql \cup \qll$, and so, we may consider it as such. Let us now consider a topmost occurrence of a relational rule $\sa$, $\idom$, $\ddom$, $\cdom$, or $\ned$ in our given derivation. By \lem~\ref{lem:permutation-1}, $\sa$ permutes above $\prdia$, $\prexi$, and $\prexii$, and since $\sa$ freely permutes above every other rule of $(\lql \cup \qll) \setminus \{\idom,\ddom\}$, the considered $\sa$ instance can be eliminated from the given derivation (the $\diar$ and $\existsr$ cases follow by \rmk~\ref{rem:diar-existsr-reachability-instances}). Similarly, by \lem~\ref{lem:permutation-2} and \rmk~\ref{rem:diar-existsr-reachability-instances}, a topmost occurrence of $\idom$, $\ddom$, or $\cdom$ can be eliminated from the given derivation. Moreover, observe that $\ned$ freely permutes above all rules of $\lql \setminus \big(\{\sa \ | \ \gpc(n,k) \in \fcset\} \cup \{\idom,\ddom\}\big)$ with the exception of $\prexi$. The only case in which $\ned$ does not freely permute above $\prexi$ is when $y$ is active in both $\prexi$ and $\ned$: 
\begin{center}
\begin{tabular}{c} 
\AxiomC{$\rel, y \in D(u), \Gamma \sar w :\phi[y/x], w :\exists x \phi, \Delta$}
\RightLabel{$\prexi$}
\UnaryInfC{$\rel, y \in D(u), \Gamma \sar w :\exists x \phi, \Delta$}
\RightLabel{$\ned$}
\UnaryInfC{$\rel, \Gamma \sar w :\exists x \phi, \Delta$}
\DisplayProof


\end{tabular}
\end{center} 
 However, the above case can be resolved by applying $\prexii$ to the top sequent, which yields the conclusion. By repeatedly eliminating topmost occurrences of relational rules as described above, 
 we eventually obtain a proof in $\qll$.
\end{proof}

\begin{theorem}[$\qll$ Soundness and Completeness]\label{thm:qll-sound-complete} 
 $\vdash w : \phi$ is derivable in $\qll$ \iffi $\phi \in \ql$.
\end{theorem}

\begin{proof} The forward direction (soundness) is shown by induction on the height of a given derivation in $\qll$. In particular, one interprets labeled sequents on models over the class $\mathcal{F}_{\mathcal{C}}$ of frames, showing that $\id$ is $\mathcal{F}_{\fcset}$-valid and that all inference rules of $\qll$ preserve validity (similar to~\cite[\thm~12.13]{NegPla11}). For the backward direction (completeness), we know that if $\phi \in \ql$, then $\vdash w : \phi$ is provable in $\lql$; hence, by \thm~\ref{thm:lab-to-ref-lab}, we may transform the derivation in $\lql$ into a derivation of $\vdash w : \phi$ in $\qll$. 
\end{proof}

%% file: nested-systems.tex
Let $\Gamma$ and $\Delta$ be multisets of $\lang$-formulae, $\vx \subset \var$ be a multiset of variables called a \emph{signature}, and $\{w_{1}, \ldots, w_{m}\} \subset \lab$ be a set of labels. A \emph{nested sequent} $\ns$ is recursively defined as:
\begin{enumerate}

\item Each \emph{flat sequent} of the form $\vx; \nant \sar \ncon$ is a nested sequent;

\item Any expression of the form $\vx; \Gamma \sar \Delta, [\nsii_{1}]_{w_{1}}, \ldots, [\nsii_{m}]_{w_{m}}$, where $\nsii_{i}$ are nested sequents for each $1 \leq i \leq m$, is a nested sequent.

\end{enumerate}
 In the context of nested sequents, we use $\Gamma$, $\Delta$, $\nantii$, and $\nconii$ to denote multisets of $\lang$-formulae, $\vx$ and $\vy$ to denote signatures (which encode domains), $\ns$ and $\nsii$ to denote nested sequents, and a semi-colon to separate signatures from multisets of $\lang$-formulae. We use $w$, $u$, $v$, $\ldots$ (occasionally annotated) to denote labels from $\lab$ and we make the simplifying assumption that every occurrence of a label in a nested sequent is unique. The incorporation of labels in our nested systems is useful as it simplifies the presentation of our reachability rules below. 

 Observe that every nested sequent of the above form encodes a \emph{tree} whose nodes are flat sequents, i.e. each nested sequent $\ns = \vx; \Gamma \sar \Delta, [\nsii_{1}]_{w_{1}}, \ldots, [\nsii_{m}]_{w_{m}}$ can be graphically depicted as a tree $tr_{w_{0}}(\ns)$ of the following form:
\begin{center}
\begin{tabular}{c c c}
\xymatrix@C+=-3.5em@R=1.5em{
		&  \overset{w_{0}}{\boxed{\vx; \nant \sar \ncon}}\ar@{->}[dl]\ar@{->}[dr] &   		\\
 tr_{w_{1}}(\nsii_{1}) & \hdots  & tr_{w_{m}}(\nsii_{m})
}
\end{tabular}
\end{center}
 We refer to a flat sequent $\vx_{i}; \nant_{i} \sar \ncon_{i}$ occurring as a node labeled with $w_{i}$ in (the tree of) a nested sequent as a \emph{$w_{i}$-component}, or as a \emph{component} more generally if we do not wish to specify its label. We note that the root $\vx; \nant \sar \ncon$ is always assumed to be associated with the label $w_{0}$, that is, $\vx; \nant \sar \ncon$ is the unique $w_{0}$-component. We use the notation $\ns\{\nsii_{1}\}_{u_{1}}\cdots\{\nsii_{m}\}_{u_{k}}$ to indicate that $\nsii_{i}$ is a $u_{i}$-component, for $1 \leq i \leq k$. 

Since nested sequents encode trees and labeled sequents encode binary graphs, it is natural to view nested sequents as types of labeled sequents. Indeed, correspondences between \emph{labeled tree sequents}~\cite{IshKik07} (i.e. labeled sequents $\rel, \Gamma \sar \Delta$ such that $\rel$ forms a (labeled) tree and every label of the sequent either occurs in $\rel$ or is identical if the $\rel = \emptyset$) and nested sequents have been established in various settings~\cite{CiaLyoRamTiu21,GorRam12,Lyo21b,Pim18}. As our nested systems are extracted from the refined labeled systems (viz. $\qll$) of the previous section, we define translations between nested and labeled tree sequents. We use these translations to transform each refined labeled calculus into a nested calculus. 
 
\begin{definition}[Translation $\ltr$] Let $\ns$ be a nested sequent and for $\vx = \{x_{1}, \ldots, x_{n}\} \subset \var$, let us define $\vx \in D(w) := x_{1} \in D(w), \ldots, x_{n} \in D(w)$. We define $\ltr(\ns) := \ltr_{w_{0}}(\ns)$ and recursively define $\ltr_{w_{0}}(\ns) := \rel, \Gamma \sar \Delta$ as follows:
\begin{enumerate}

\item $\ltr_{v}(\vx, \nant\sar \ncon) := \vx \in D(v), v : \nant \sar v : \Delta$, and

\item $\ltr_{v}(\vx; \nant \sar \ncon, \bl \nsii_{1} \br_{u_{1}}, \ldots, \bl \nsii_{k} \br_{u_{k}}) :=$
$$
\ltr_{v}(\vx; \nant \sar \ncon) \seqcomp (vRu_{1}, \ldots, vRu_{k} \sar \empseq) \seqcomp \ltr_{u_{1}}(\nsii_{1}) \seqcomp \cdots \seqcomp \ltr_{u_{k}}(\nsii_{k}).
$$
\end{enumerate}
\end{definition}

\begin{definition}[Translation $\ntr$] Let $\Lambda := \rel, \mathcal{D}, \Gamma \sar \Delta$ be a labeled tree sequent with $\rel$ containing only relational atoms of the form $uRv$, $\mathcal{D}$ containing only domain atoms of the form $x \in D(u)$, and let $w$ be the root of $\Lambda$. We define $\Lambda_{1} \subseteq \Lambda$ \ifandonlyif there exists a labeled tree sequent $\Lambda_{2}$ such that $\Lambda = \Lambda_{1} \seqcomp \Lambda_{2}$ and where if $u$ is the root of $\Lambda_{1}$, then $\Lambda_{2}$ contains no domain atoms or labeled formulae associated with $u$. Let us further define $\Lambda_{u}$ to be the unique labeled tree sequent rooted at the label $u$ such that $\Lambda_{u} \subseteq \Lambda$. We define $\ntr(\Lambda) := \ntr_{u}(\Lambda)$ recursively on the tree structure of $\Lambda$:
\begin{enumerate}

\item if $\rel = \empseq$, then $\ntr_{v}(\Lambda) := \{x \ | \ (x \in D(v)) \in \mathcal{D}\}; (\Gamma \restriction v) \sar (\Delta \restriction v)$, and

\item if $vRz_{1}, \ldots vRz_{n}$ are all of the relational atoms occurring in the input sequent which have the form $vRz$, then
$$
\ntr_{v}(\Lambda) := \{x \ | \ (x \in D(v)) \in \mathcal{D}\}; (\Gamma \restriction v) \sar (\Delta \restriction v), \bl \ntr_{z_{1}}(\Lambda_{z_{1}}) \br, \ldots, \bl \ntr_{z_{n}}(\Lambda_{z_{n}})].
$$

\end{enumerate}
\end{definition}

\begin{example} We let $\ns := \emptyset; \exists x p(x) \sar q \lor r, \bl y; p \sar q(y), \bl y, z; \dia p \sar \dia q \br_{u} \br_{v}$ and show the output labeled sequent under the translation $\ltr$. We have
$$
\ltr(\ns) = \rel, w_{0} : \exists x p(x), v : p, u : \dia p \sar w_{0} : q \lor r, v : q(y), u : \dia q
$$
where $\rel = w_{0}Rv,vRu, y \in D(v), y \in D(u), z \in D(u)$. Moreover, $\ntr(\ltr(\ns)) = \ns$.
\end{example}


Before we define our nested systems, we define \emph{propagation graphs} and related concepts in the setting of nested sequents. Such notions easily transfer from the labeled setting since every nested sequent can be viewed as a labeled sequent under the $\ltr$ translation.

\begin{definition}[Propagation Graph of a Nested Sequent] Let $\ns$ be a nested sequent and let $\ltr(\ns) := \rel, \Gamma \sar \Delta$. We define the \emph{propagation graph} of $\ns$ to be $\prgr{\ns} := \prgr{\rel}$ and we note that propagation paths, their converses, and strings thereof are defined as in \dfn~\ref{def:propagation-path}.
\end{definition}

A uniform presentation of all nested systems is provided in \fig~\ref{fig:nested-calculi}. We remark that $\nqk := \mathsf{NQK}(\emptyset)$ and that $\prdia$, $\prexi$, and $\prexii$ constitute our reachability rules. The side conditions for such rules are identical to the side conditions given in the previous section, i.e. to apply the rule top-down (or, bottom-up) one takes the premise (conclusion, resp.) $\ns$, obtains $\rel$ via the translation $\ltr(\ns) = \rel,\Gamma \sar \Delta$, and then checks 
 if the side conditions hold relative to $\prgr{\rel}$.

\begin{figure}[t]

\begin{center}
\begin{tabular}{c c c}
\AxiomC{}
\RightLabel{$\ax$}
\UnaryInfC{$\ns \hol \vx; \nant, p(\vec{x}) \sar p(\vec{x}), \ncon \hor_{w}$}
\DisplayProof

&

\AxiomC{$\ns \hol \vx; \nant \sar \phi, \ncon \hor_{w}  $}
\RightLabel{$\negl$}
\UnaryInfC{$\ns \hol \vx; \nant, \neg \phi \sar \ncon \hor_{w}  $}
\DisplayProof

&

\AxiomC{$\ns \hol \vx; \nant, \phi \sar \ncon \hor_{w}  $}
\RightLabel{$\negr$}
\UnaryInfC{$\ns \hol \vx; \nant \sar \neg \phi, \ncon \hor_{w}  $}
\DisplayProof
\end{tabular}
\end{center}

\begin{center}
\begin{tabular}{c c c}
\AxiomC{$\ns \hol \vx; \nant, \phi \sar \ncon \hor_{w}  $}
\AxiomC{$\ns \hol \vx; \nant, \psi \sar \ncon \hor_{w}  $}
\RightLabel{$\disl$}
\BinaryInfC{$\ns \hol \vx; \nant, \phi \lor \psi \sar \ncon \hor_{w}$}
\DisplayProof

&

\AxiomC{$\ns \hol \vx; \nant \sar \phi, \psi, \ncon \hor_{w}  $}
\RightLabel{$\disr$}
\UnaryInfC{$\ns \hol \vx; \nant \sar \phi \lor \psi, \ncon \hor_{w}  $}
\DisplayProof
\end{tabular}
\end{center}

\begin{center}
\begin{tabular}{c c}
\AxiomC{$\ns \hol \vx; \nant \sar \phi[y/x], \exists x \phi, \ncon \hor_{w}$}
\RightLabel{$\prexi^{\dag_{2}(\fcset)}$}
\UnaryInfC{$\ns \hol \vx; \nant\sar \exists x \phi, \ncon \hor_{w}$}
\DisplayProof

&

\AxiomC{$\ns \hol \vy,y; \nantii \sar \nconii \hor_{u} \hol \vx; \nant \sar \phi[y/x], \exists x \phi, \ncon \hor_{w}$}
\RightLabel{$\prexii^{\dag_{3}(\fcset)}$}
\UnaryInfC{$\ns \hol \vy; \nantii \sar \nconii \hor_{u} \hol \vx; \nant \sar \exists x \phi, \ncon \hor_{w}$}
\DisplayProof
\end{tabular}
\end{center}

\begin{center}
\begin{tabular}{c c}
\AxiomC{$\ns \hol \vx; \nant \sar \ncon, [\emptyset; \phi \sar \emptyset]_{u} \hor_{w} $}
\RightLabel{$\dial$}
\UnaryInfC{$\ns \hol \vx; \nant, \dia \phi \sar \ncon \hor_{w} $}
\DisplayProof

&

\AxiomC{$\ns \hol \vx; \nant \sar \Diamond \phi, \ncon \hor_{w} \hol \vy; \nantii \sar \phi, \nconii \hor_{u}$}
\RightLabel{$\prdia^{\dag_{1}(\fcset)}$}
\UnaryInfC{$\ns \hol \vx; \nant \sar \Diamond \phi, \ncon \hor_{w} \hol \vy; \nantii  \sar \nconii \hor_{u}$}
\DisplayProof
\end{tabular}
\end{center}

\begin{center}
\begin{tabular}{c c c}
\AxiomC{$\ns \hol \vx, \vec{x}; \nant, p(\vec{x}) \sar \ncon \hor_{w}$}
\RightLabel{$\domr$}
\UnaryInfC{$\ns \hol \vx; \nant, p(\vec{x}) \sar \ncon \hor_{w}$}
\DisplayProof

&

\AxiomC{$\ns \hol \vx, y; \nant, \phi[y/x] \sar \ncon \hor_{w} $}
\RightLabel{$\existsl^{\dag}$}
\UnaryInfC{$\ns \hol \vx; \nant, \exists x \phi \sar \ncon \hor_{w} $}
\DisplayProof

&

\AxiomC{$\ns \hol \vx; \nant \sar \ncon, [ \emptyset; \emptyset \sar \emptyset]_{u} \hor_{w}$}
\RightLabel{$\D$}
\UnaryInfC{$\ns \hol \vx; \nant \sar \ncon \hor_{w}$}
\DisplayProof
\end{tabular}
\end{center}

\caption{The nested system $\nql$, which includes $\D$ \iffi $\serc \in \fcset$. The side condition $\dag$ states that the rule is applicable only if $y$ is fresh; see \dfn~\ref{def:side-conditions} for the interpretation of the side conditions $\dag_{1}(\fcset) -\dag_{3}(\fcset)$, which are functions of $\fcset$.}
\label{fig:nested-calculi}
\end{figure}

\begin{definition}[Labeled Tree Derivation~\cite{Lyo21b}]
We define a \emph{labeled tree derivation} to be a proof containing only labeled tree sequents. We say that a labeled tree derivation has the \emph{fixed root property} \ifandonlyif every labeled sequent in the derivation has the same root.
\end{definition}

\begin{lemma}\label{lem:labeled-tree-derivations}
Every proof in $\qll$ of a labeled tree sequent is a labeled tree proof with the fixed root property.
\end{lemma}

\begin{proof}
The lemma follows from the fact that if any rule of $\qll$ is applied bottom-up to a labeled tree sequent, then each premise is a labeled tree sequent with the same root. 
\end{proof}

\begin{theorem}[$\nql$ Soundness and Completeness]\label{thm:soundness} 
 $\emptyset; \emptyset \sar \phi$ is derivable in $\nql$ \iffi $\phi \in \ql$. 
\end{theorem}

\begin{proof} Follows from \thm~\ref{thm:qll-sound-complete} and \lem~\ref{lem:labeled-tree-derivations} using the $\ltr$ and $\ntr$ translations.
\end{proof}

 We stress that any sound formulation of a cut rule is redundant in $\nql$ as each nested calculus is complete without the inclusion of such a rule, i.e. our nested systems are cut-free.






%% file: conclusion.tex
 In this work, we unified first-order modal logics in a single nested sequent framework by means of reachability rules. These rules, which may propagate formulae and/or check for data along paths within a nested sequent, permit one to toggle between the various logics by parameterizing rules with distinct semi-Thue systems. We used such rules to provide 
 the first cut-free nested sequent systems for first-order modal logics.
 
 In future work, we aim to investigate how reachability rules can be leveraged to provide nested systems for broader classes of logics. This can be accomplished in multiple ways. For instance, we could investigate the effect of parameterizing reachability rules with alternative semi-Thue systems, or through the introduction of new logical connectives (e.g. converse modalities). Second, we aim to explore the formulation of reachability rules for alternative quantifiers; e.g. Miller and Tiu's nabla quantifier $\nabla$~\cite{MilTiu05} or the Gabbay-Pitts $\new$ quantifier~\cite{GabPit99}. Moreover, we aim to investigate the properties of our nested systems such as syntactic cut-elimination, 
 invertibility of rules, and the admissibility of structural rules (e.g. contraction and weakening), which are conjectured to hold. 
  Last, such systems may prove well-suited for identifying decidable fragments of first-order modal (and related) logics by means of proof-search algorithms. 